\documentclass{llncs}
\usepackage{amsmath, amssymb, amsfonts}
\usepackage{graphicx}
\usepackage{tikz}
\usetikzlibrary{decorations.pathreplacing}

\newtheorem{defn}[theorem]{Definition}
\newtheorem{prop}[theorem]{Proposition}
\newtheorem{lem}[theorem]{Lemma}
\newtheorem{cor}[theorem]{Corollary}

\begin{document}

\title{Betweenness centrality profiles in trees}


\author{Benjamin Fish\inst{1}\thanks{Contact author:  bfish3@uic.edu.  Benjamin Fish was funded in part by Army Research Office grant \#66497-NS.} \and
Rahul Kushwaha\inst{1} \and
Gy\"orgy Tur\'an\inst{1,2}}

\institute{
Department of Mathematics, Statistics, and Computer Science,\\ University of Illinois at Chicago
\and
MTA-SZTE Research Group on Artificial Intelligence, Szeged
}

\maketitle

\begin{abstract}
{Betweenness centrality of a vertex in a graph measures the fraction of shortest paths going through the vertex.
This is a basic notion for determining the importance of a vertex in a network.
The $k$-betweenness centrality of a vertex is defined similarly, but only considers shortest paths of length at most $k$. The sequence of $k$-betweenness centralities
for all possible values of $k$ forms the betweenness centrality profile of a vertex.
We study properties of betweenness centrality profiles in trees.

We show that for scale-free random trees, for fixed $k$, the expectation of $k$-betweenness centrality strictly decreases as the index of the vertex increases. We also analyze worst-case properties of profiles in terms of the distance of profiles from being monotone, and the
number of times pairs of profiles can cross.  This is related to whether $k$-betweenness centrality, for small values of $k$, may be used instead of having to consider all shortest paths.  Bounds are given that are optimal in order of magnitude.
We also present some experimental results for scale-free random trees.
}
\end{abstract}

\section{Introduction}
The study of large networks uses many graph parameters measuring properties of vertices, groups of vertices, or the whole network.
Some, like connectivity, are classical concepts. Others, like the Bavelas centrality index~\cite{bavelas1948mathematical}, come
from applications in sociology and psychology. Still others, like $k$-betweenness centrality, to be discussed in this paper,
have been introduced recently, partly motivated by efficiency considerations. Several of these notions have mostly been studied
experimentally, but they also raise many questions for theoretical study. We consider some problems of this kind.

A basic property of a vertex in a network is its \emph{centrality}. Possible definitions are surveyed by Borgatti and Everett~\cite{borgatti2006graph}
and Brandes and Erlebach~\cite{brandes2005network}.
One intuition is that a vertex is central if it lies on many shortest paths.
A popular formalization, introduced by Anthonisse~\cite{anthonisse1971rush} and Freeman~\cite{freeman1977set},
is \emph{betweenness centrality}, measuring the fraction of shortest paths going through the vertex.
Betweenness centrality has many applications, for example in examining ``trade routes between medieval cities'' where ``the cities with high betweenness centrality have opportunities for amassing wealth and exerting control that other cities would not have''~\cite{borgatti2006graph,pitts1979medieval}.
Even though betweenness centrality can be computed in polynomial time using all-pairs shortest paths algorithms, for large networks
it is important to reduce the running times. Algorithms are discussed,
for example, in Brandes~\cite{brandes2001faster}, Eppstein and Wang~\cite{EppWan-SODA-01}, and Riondato and Kornaropoluos~\cite{riondato2014fast}.

In this paper, we will focus on the behavior of betweenness centrality in trees.  The distribution of path lengths in trees has been studied before, for example by Faudree et al.~\cite{faudree1973theory}, Chaudhary and Gordon~\cite{chaudhary1991tutte}, Gordon and McDonnell~\cite{gordon1995trees}, and Dankelmann~\cite{dankelmann2011distance}.
Riondato and Kornaropoulos~\cite{riondato2014fast}
give bounds on the VC-dimension of shortest paths of length at most $k$ in arbitrary graphs, which is relevant for determining the sample size needed by
approximation algorithms.

The notion of $k$-\emph{betweenness centrality}, introduced by Borgatti and Everett~\cite{borgatti2006graph},
differs from betweenness centrality in that shortest paths are restricted to length at most $k$.
 In applications, shorter paths may be more important in determining the centrality of a vertex.  Shorter paths may also give more information about the centrality of vertices, e.g., when two vertices have the same centrality when not restricted to shorter paths.  In addition, considering only short paths can have computational advantages as well:  considering only short paths may increase the locality of the computation, which can make computing faster and easier.

An experimental study of $k$-betweenness centrality, and other bounded versions of centrality measures,
is given by Pfeffer and Carley~\cite{pfeffer2012k}.
They consider the following general requirements:
1) computing the bounded versions should be similar to the unbounded version, 2) the bounded versions should approximate
the unbounded one and 3) the bounded versions should converge to the unbounded one when the bound parameter grows.
As an illustration, \cite{pfeffer2012k} shows $k$-betweenness centralities in a tree displaying the required properties.
Approximation and convergence are used in an informal sense in \cite{pfeffer2012k}; in their example
$k$-betweenness centralities converge to betweenness centrality monotonically from below.
Experiments of \cite{pfeffer2012k} on random graphs also support the claim that $k$-betweenness centrality has ``good'' properties.

The experiments of~\cite{pfeffer2012k} motivate a theoretical study of $k$-betweenness centrality.  We ask to what degree does the bounded version of betweenness centrality (i.e.~$k$-betweenness centrality) approximate the unbounded version, and more generally, to what degree does $k$-betweenness centrality approximate $\ell$-betweenness centrality for $k$ smaller than $\ell$?  If smaller values of $k$ approximate the larger values well, then we may safely consider only short paths in the graph.

We consider the maximum size of the gap between the $k$-betweenness centrality and the betweenness centrality.  However, even if the $k$-betweenness centrality approximates the betweenness centrality poorly, in some applications this might not matter.  For example, we might wish to rank vertices by their centrality~\cite{riondato2014fast}, begging the following question:  Is the vertex ranking by betweenness centrality preserved under $k$-betweenness centrality?  Failing that, will the $k$-betweenness centrality increase or decrease as $k$ increases?

Let $BC_k(v)$ be the $k$-betweenness centrality of a vertex $v$ in a graph $G$.
The  \emph{betweenness centrality profile}, or, for simplicity, the \emph{profile} of $v$ is the sequence ${\cal B}(v) = (BC_2(v), BC_3(v), \ldots,  BC_{d}(v))$,
where $d$ is the diameter of $G$.
To answer the above questions we introduce the notions of \emph{dips} and \emph{crossings} in profiles.
Informally, a dip in a profile is non-monotonic segment of the profile, and a crossing between two profiles is a segment where the comparison of the two profiles switches.  If there are many crossings, than rankings on $k$-betweenness centrality are unreliable, and if there are many dips, then estimates as to whether the $k$-betweenness centrality will increase or decrease as $k$ increases will be unreliable.


The $k$-betweenness centralities are properties which, with increasing $k$, transition from being local to being global.
The general, informal conclusion of \cite{pfeffer2012k} is that this transition from local to global is smooth.
Properties of profiles provide quantifiable information about this transition. In what follows we refer to results indicating
a smooth transition as ``positive'', and to results showing pathologies in this transition as ``negative''.
For instance, profiles in a path have no dips and do not cross, indicating a completely smooth transition.
Trees with many dips and crossings represent a non-smooth transition.


In this paper we consider probabilistic and worst-case properties of profiles. We restrict our attention to trees.
Trees are a natural first case to consider as there is exactly one shortest path between any two vertices.

We consider \emph{preferential attachment}, or \emph{scale-free}, random trees~\cite{barabasi1999emergence}.  We consider this model due to its popularity and its ability to model vertices that act as `hubs,' those vertices that are much more central than the rest of the vertices.
Scale-free trees were studied in the 1990's under the name \emph{plane-oriented recursive trees} (see, e.g., Mahmoud~\cite{mahmoud1992distances} and Smythe and Mahmoud~\cite{smythe1995survey}). Results on path lengths in scale-free trees are given by Bollob\'as and Riordan~\cite{bollobas2004shortest} (see also Szab\'o et al.~\cite{szabo2002shortest}).
Recent work includes studies of the influence of the seed graph on the limiting distribution of the maximum degree
(Bubeck et al.~\cite{bubeck14seed}), algorithms identifying the root node of a scale-free tree (Bubeck et al.~\cite{bubeck14adam}), and
 the persistence of a centroid (a vertex minimizing the maximal subtree when chosen as the root) in a scale-free tree (Jog and Loh~\cite{jog16tree}).

\subsection{Results}
We give a positive result for scale-free trees.
We show that for fixed $k$, the expectation of $k$-betweenness centrality strictly decreases as the index of the vertex increases,
i.e., for every $k$ and vertices $v < w$ it holds that $\mathbb{E}[BC_k(v)] > \mathbb{E}[BC_k(w)]$.
This result can be viewed as giving evidence that $k$-betweenness centrality is ``well-behaved'' for random trees.
The proof is based on a formula of Bollob\'as and Riordan~\cite{bornholdt} (see also Bollob\'as and Riordan~\cite{bollobas2004shortest}) for the probability of
the presence of a fixed subgraph in a scale-free tree.
Proving stronger positive results on the distribution of $k$-betweeness centralities, even for trees, seems to be an interesting open problem.
Extending the results to general graphs may not be easy; as far as we know, the results of Bollob\'as and Riordan~\cite{bollobas2004shortest} have not been extended to the general case either.

We also prove negative results by studying worst-case behavior: how many dips and crossings can profiles have?  We consider worst-case results in order to bound the behavior of all profiles.

We also make some simple observations on the worst-case approximation properties of betweenness centrality by $k$-betweenness centralities (Proposition~\ref{pro:appx}).
It is noted that there are trees where the $k$-betweenness centralities of a vertex can be much smaller, resp., much larger, than the betweenness centrality of the same vertex, holding for all possible values of $k$, resp., for $k$ up to a small constant multiple of the diameter.
It follows that there is no constant-factor approximation guarantee when using $k$-betweenness centrality to approximate betweenness centrality.

As the diameter is always an upper bound to the number of dips and crossings, we consider proving lower bounds for the number of dips and crossings in terms of this quantity.
Theorem~\ref{thm:dips} shows that there are trees of diameter $d$ with some vertex
having $\Omega(d)$ dips in its profile.
The construction is simple, but the choice of the parameters and the analysis requires careful computation as the dips are small.
 Theorem~\ref{thm:crossings} gives trees of diameter $d$ with a pair of vertices such that their profiles cross each other $\Omega(d)$ many times.
In particular, it follows from these results that the worst-case behavior of dips and crossings for general graphs can be achieved with trees already.

The paper is structured as follows. In Section~\ref{sec:prel} we give the basic definitions and in Section~\ref{sec:paths} we discuss the case of paths.
Approximation is considered in Section~\ref{sec:appr}.  Scale-free trees are discussed in Section~\ref{sec:sca}.  Then lower bounds on the number of dips and crossings are given in Sections~\ref{sec:dip},
resp., \ref{sec:cro}.  Finally, experimental results are mentioned in Section~\ref{sec:expe}.


\section{Preliminaries} \label{sec:prel}

We consider undirected graphs with unweighted edges. The length of a path is the number of its edges.
The degree of $s$ is denoted by $deg(s)$ and the
distance between $s$ and $t$ is denoted by $\ell(s,t)$. 
The diameter of a graph, denoted by $d$, is the maximum of the distances $\ell(s,t)$.  The set of vertices having distance at most $k$ from $v$,  \emph{not including} $v$, is denoted by $N_v^k$.

Let $\delta_{st}$ be the number of distinct shortest paths between $s$ and $t$. For $v \neq s, t$,
let $\delta_{st}(v)$ be the number of distinct shortest paths between $s$ and $t$ that pass through $v$.
The normalized betweenness centrality of a vertex $v$ is
\[ BC(v) =  \frac{2}{n (n - 1) - 2 |E|} \cdot \sum_{s<t, \,\, 2 \le \ell(s,t)} \frac{\delta_{st}(v)}{\delta_{st}}.\]

Normalization (division by the number terms in the sum, i.e., the number of nonadjacent pairs) ensures that betweenness centrality is $[0,1]$-valued.  (This version of the definition is undefined when the graph is complete, so we define the betweenness centrality of any vertex in the complete graph as $0$.)
\begin{defn}
The normalized $k$-betweenness centrality of a vertex $v$ is
\[ BC_k(v) = \frac{2}{\sum_{i}(|N_i^k| - deg(i))} \cdot \sum_{s<t, \,\, 2 \le \ell(s,t) \le k} \frac{\delta_{st}(v)}{\delta_{st}}.\]
\end{defn}
Here, similarly to the general case, one normalizes by dividing with the number of terms in the sum, i.e.,
the number of pairs of vertices $s,t$ where $2 \le \ell(s,t) \le k$.  In this paper, we drop the term ``normalized'' and simply refer to it as $k$-betweenness centrality.
If $k$ is the diameter then the $k$-betweenness centrality of a vertex is its betweenness centrality.

If the graph is a tree then $BC_k(v)$ is simply the fraction of shortest paths of length between 2 and $k$ which contain $v$ as an interior vertex (the interior vertices of a path are all vertices that are not endpoints).
For example, the center of a star has $2$-betweenness centrality 1.

We refer to the number of shortest paths of length between 2 and $k$ containing $v$ as an interior vertex as $P_k(v)$, and the total number of paths of length between 2 and $k$ as $P_k$.  Also, we refer to the number of shortest paths of length \emph{exactly} $k$ containing $v$ as an interior vertex as $p_k(v)$, and similarly the total number of shortest paths of length \emph{exactly} $k$ as $p_k$.

\begin{defn}
The profile of a vertex $v\in G$ is  \[ {\cal B}(v) = (BC_2(v), BC_3(v), \ldots,  BC_{d}(v)).\]
\end{defn}

We use the following definition for measuring non-monotonicity.

\begin{defn} (Dip in a profile)
Given a profile ${\cal B}(v)$,
an interval \[ BC_i(v),BC_{i+1}(v),\ldots,BC_{j}(v) \] is a \emph{dip} if it can be partitioned into two intervals, the first of which is monotonically non-increasing and the second of which is monotonically non-decreasing, but the entire interval is neither monotonically non-increasing nor non-decreasing.  A profile has $q$ dips if there are $q$ disjoint intervals that are dips.
\end{defn}

Note that if a profile has $q$ dips then at least $q$ entries must be changed to make it monotonic (increasing or decreasing), thus in this sense it is
$q$-far from being monotonic.
The following definition is about comparing two profiles.

\begin{defn} (Crossing)
Profiles ${\cal B}(u)$ and ${\cal B}(v)$ have $m$ \emph{crossings} if there exist indices $i_1 <  \ldots < i_{m+1}$ ($2\le i_s \le d$) where for all such indices, $BC_{i_s}(u) \neq BC_{i_s}(v)$ and the values alternate order (so if $BC_{i_s}(u) < BC_{i_s}(v)$ then $BC_{i_{s+1}}(u) > BC_{i_{s+1}}(v)$ and vice versa).
\end{defn}

As an illustration, consider the example in Figure~\ref{fig:claw_example} below.

\begin{figure}[h]
\centering
\tikzstyle{vertex}=[circle,fill=black,minimum size=10pt,inner sep=0pt]
\tikzstyle{edge}=[-,color=black,line width=1.5]
\tikzstyle{edgelabel}=[color=black]
\begin{tikzpicture}[scale=1, auto]
    \node[vertex] (0) at (0,0) {};
    \node[vertex] (1) at (1,0) [label=above:$v_1$] {}
        edge[edge] node {} (0);
    \node[vertex] (2) at (2,0) [label=above:$v_2$] {}
        edge[edge] node {} (1);
    \node[vertex] (3) at (3,0) [label=above:$v_3$] {}
        edge[edge] node {} (2);
    \node[vertex] (4) at (4,0) [label=above:$v_4$] {}
        edge[edge] node {} (3);
    \node[vertex] (5) at (5,1) [label=above:$v_5$] {}
        edge[edge] node {} (4);
    \node[vertex] (6) at (5,0) [label=above:$v_6$] {}
        edge[edge] node {} (4);
    \node[vertex] (7) at (5,-1) [label=above:$v_7$] {}
        edge[edge] node {} (4);
    \node[vertex] (8) at (6,1)  {}
        edge[edge] node {} (5);
    \node[vertex] (9) at (6,0){}
        edge[edge] node {} (6);
    \node[vertex] (10) at (6,-1) {}
        edge[edge] node {} (7);
\end{tikzpicture}
\caption{Example tree}
\label{fig:claw_example}
\end{figure}
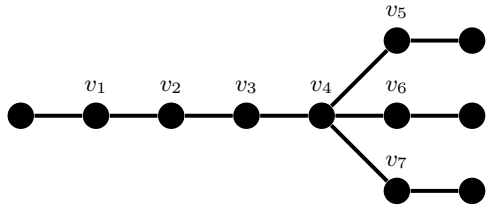

\begin{figure}[h]
\begin{center}
\includegraphics[scale=0.45]{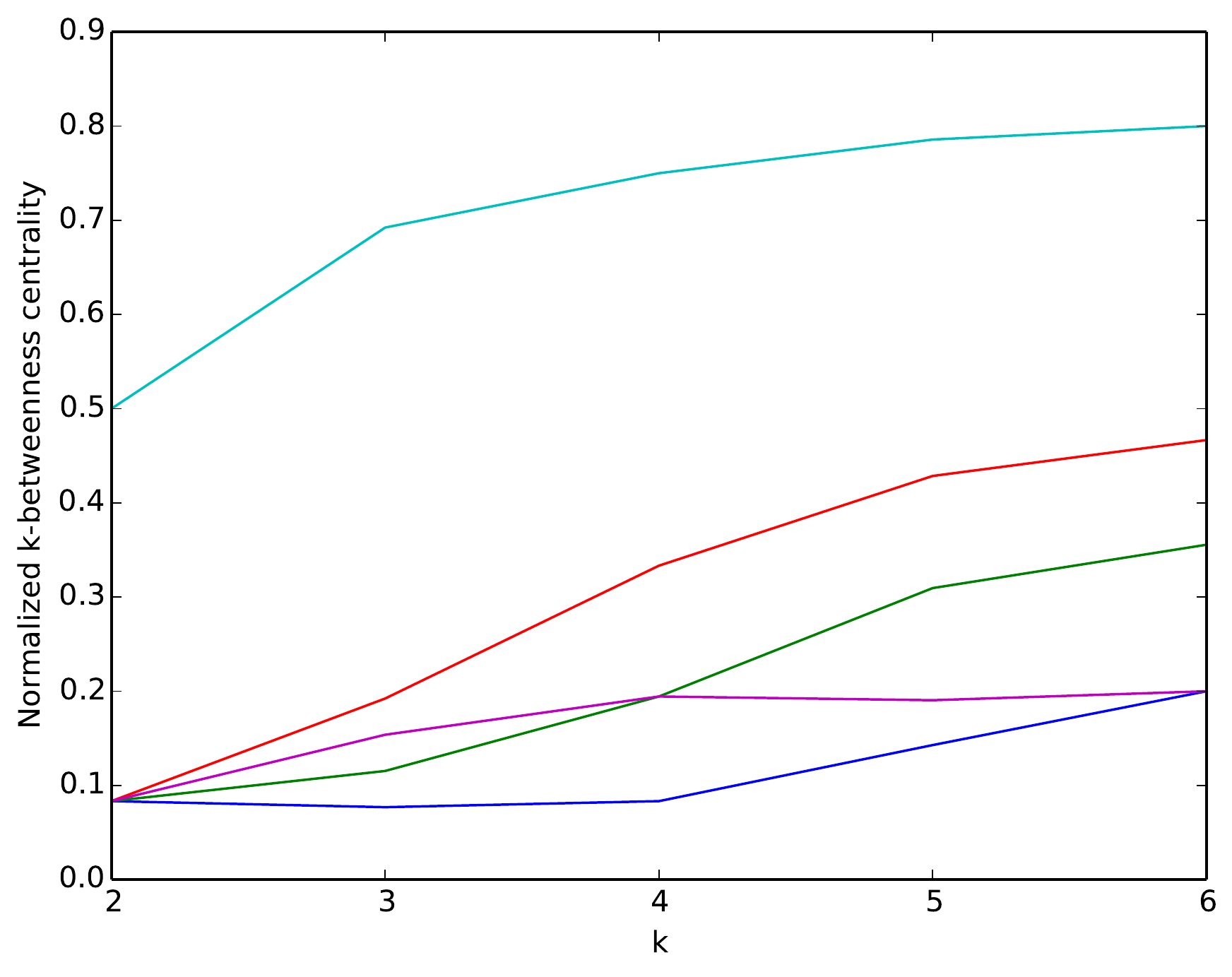}
\end{center}
\caption{The profiles of vertices $v_1$, $v_2$, $v_3$, $v_4$, and $v_5$ from Figure~\ref{fig:claw_example},
in the order $v_4, v_3, v_2, v_5, v_1$ from top to bottom.}
\label{fig:example_profiles}
\end{figure}

The diameter of this tree is 6, so a profile consists of the $k$-betweenness centralities for $k=2$ through $k=6$.  The
profile of the four leaves is the all zero vector as no paths contain them as an internal vertex.
The profiles of $v_5$, $v_6$, and $v_7$ are the same, so we omit the latter two and plot the other five profiles in Figure~\ref{fig:example_profiles}.
Note that the profiles of $v_1$ and $v_5$ are not monotonic, but the rest of the vertices do have monotonic profiles.
It is also interesting to note that $BC(v_1)=BC(v_5)$, but $BC_k(v_1) \le BC_k(v_5)$ for all $k$.
If short paths are considered more relevant then it may be reasonable to consider $v_5$ as having a higher centrality than $v_1$.
Thus, in this case the profiles give useful extra information in addition to betweenness centrality.
On the other hand, the profiles of $v_2$ and $v_5$ cross at $k=4$. In this case, based on the profiles, one might qualify the judgement using betweenness centrality alone that $v_2$ is more central than $v_5$.

\subsection{Paths} \label{sec:paths}
We begin the discussion of $k$-betweenness centralities of trees by considering the simplest nontrivial case, that of paths.
It is relatively easy to find examples of trees where the vertices' profiles behave well.
As a first example, we consider simple paths, where the relevant quantities are easily calculated explicitly.  The profiles of vertices in simple paths neither have dips nor crossings:
\begin{prop}\label{prop:paths}
In every simple path, profiles are monotonically increasing as $k$ increases and pairs of paths do not cross.
\end{prop}

\begin{proof}
A path of length $n$ has $n-\ell+1$ subpaths of length $\ell$, i.e., $p_\ell = n-\ell+1$.  Label the vertices $0$ through $n-1$ along the path.
A subpath of length $\ell$ starting at vertex $j<i$ includes $i$ if and only if its other end is past the vertex $i$ but it is not past the last vertex, i.e., if $i-j+1\le \ell \le n-j$, or equivalently $i+1-\ell \le j < n-\ell+1$.  It then follows that if $\ell \le i$, there are $\ell-1$ subpaths including $i$ of length $\ell$,
i.e., $p_\ell(i) = \ell - 1$.  Similarly, if $i < \ell \le n+1-i$  then $p_\ell(i) = i$. Finally, if $\ell > n+1-i$ then $p_\ell(i) = n-\ell+1$.
Therefore $p_\ell(i) / p_\ell$ is monotonically increasing in $\ell$ and so
\[ BC_k(i) = \frac{\sum_{\ell = 2}^k p_\ell(i)}{\sum_{\ell = 2}^k p_\ell} \]
is monotonically increasing in $k$. It also follows that $p_\ell(i)$ is monotonically increasing in $i$, so $P_k(i)$ is also monotonically increasing in $i$.
Thus profiles do not cross.
\end{proof}

For the sake of completeness, we now give the exact values of the $k$-betweenness centralities, which may be proved by examining the proof above.  It also may be viewed as a warm-up for the direct calculations of $k$-betweenness centralities in Theorems~\ref{thm:dips} and~\ref{thm:crossings}.


\begin{cor}\label{lem:simple_path_form}
Let $i$ be the $i$th vertex of a simple path of length $n$, for $0\le i \le \lfloor\frac{n}{2}\rfloor$.  For $2\le k \le n$, the number of paths of length
at most $k$ containing vertex $i$ is
\[  P_k(i) =
\begin{cases}
    \binom{k}{2}& \text{if } k \le i\\
    \binom{i}{2} +i(k-i)& \text{if } i < k \le n+1-i\\
    \binom{i}{2} + i(n-2i+1) +(n+1)(k-n+i-1)\\
    + \binom{n+2-i}{2} - \binom{k+1}{2} & \text{otherwise}
\end{cases}
\]
and the $k$-betweenness centrality of a path of length $n$ at vertex $i$ is
\begin{equation*} \label{eq:pat}
BC_k(i) = \frac{P_k(i)}{k(n+1)-n-\binom{k+1}{2}}.
\end{equation*}
\end{cor}






\subsection{Approximation} \label{sec:appr}

Lower bounds are proven by constructing sequences of trees where profiles behave badly.
As a first simple example we consider approximation properties.

A \emph{broom} $H_{mn}$ consists of a path of length $m$ with $n$
 leaves attached to an endpoint. A \emph{double broom} $K_{mn}$ consists of a path of length $m$ with $n$
leaves attached to both endpoints. Double brooms and the middle vertex show that $k$-betweenness centralities can be much smaller
than betweenness centrality for every $k$ smaller than the diameter. Brooms and the center vertex show that $k$-betweenness centralities can be much larger than betweenness centrality, as long as $k$ is at most a small constant fraction of the diameter.
Thus $k$-betweenness centralities in the corresponding ranges cannot be guaranteed to approximate betweenness centrality within a constant factor.

\begin{prop}\label{pro:appx} \emph{a)}
For every $\epsilon > 0$ there is a $d_0$ such that for every $d \ge d_0$ there is
a tree $T$ of diameter $d$ and a vertex $v$ in $T$ such that for all $k<d$ it holds that \[\frac{BC_k(v)}{BC_d(v)} < \epsilon.\]

\emph{b)} For every $C > 0$ there is $d_0$ and $\epsilon > 0$ such that for every $d \ge d_0$ there is a tree $T$ of diameter $d$ and a vertex $v$ in $T$ such that
for every $k \le \epsilon \cdot d$ it holds that
\[\frac{BC_k(v)}{BC_d(v)} > C. \]

\end{prop}

\begin{proof}[Proposition~\ref{pro:appx} (a)]
Consider the double broom $K_{mn}$, where $m$ is even, and let $v$ be the middle vertex on the path of length $m$. The diameter is $d = m+2$.

If $k < d$ then no paths between two leaves have length at most $k$ and go through $v$, no matter whether the leaves are attached to the same endpoint of the central path of the broom, or to different endpoints. Thus counting pairs of leaves attached to the same, resp. to different, endpoints, we get
\[ P_k(v) \le {{m + 2n} \choose 2} - n (n-1) - n^2 \le 2 m n + m^2 +n. \]
The total number of paths of length at most $k$ include paths between neighbors of same endpoint thus $P_k \ge n^2$.
Paths between leaves attached to different endpoints do go through $v$ thus $P_d(v) \ge n^2$.
Thus for some constant $K$ depending on $\epsilon$ if $n \ge K \cdot m$ then
\[ \frac{BC_k(v)}{BC_d(v)} = \frac{P_k(v) \cdot P_d}{P_k \cdot P_d(v)} \le \frac{(2 m n + m^2 +n) \cdot (m + 2n)^2}{n^4} < \epsilon.\]
\end{proof}

\begin{figure}
\centering
\tikzstyle{vertex}=[circle,fill=black,minimum size=10pt,inner sep=0pt]
\tikzstyle{edge}=[-,color=black,line width=1.5]
\tikzstyle{edgelabel}=[color=black]
\begin{tikzpicture}[scale=1, auto]
    \node[vertex] (a) at (0,1) {};
    \node[vertex] (b) at (1,1) {}
        edge[edge] node {} (a);
    \node[vertex] (c) at (2,1) {}
        edge[edge] node {} (b);
    \node[vertex] (v) at (4,1) [label=below:$v$] {};
    \node[vertex] (x) at (5,2) {}
    	edge[edge] node {} (v);
    \node[vertex] (y) at (5,3) {}
    	edge[edge] node {} (v);
    \node[vertex] (z) at (5,0) {}
    	edge[edge] node {} (v);

    \draw [decorate,decoration={brace,amplitude=15pt},xshift=0pt,yshift=8pt, line width = 1]
    (0,1) -- (4,1) node [black,midway,yshift=15pt] {$m$};

     \draw [decorate,decoration={brace,amplitude=15pt,mirror},xshift=8pt,yshift=0pt, line width = 1]
    (5,0) -- (5,3) node [black,midway,xshift=30pt] {$n$};

    \node[draw=none] at (3,1) {$\cdots$};
    \node[draw=none] at (5,1) {$\vdots$};

\end{tikzpicture}
\caption{$H_{mn}$}
\label{fig:claw}
\end{figure}

\begin{proof}[Proposition~\ref{pro:appx} (b)]
Consider the broom $H_{mn}$, shown in Figure~\ref{fig:claw}, and let $v$ be the endpoint with the leaves attached.
The diameter is $d = m+1$. Now we can estimate the relevant quantities as follows.
Looking at the neighbors of $v$ it follows that $P_k(v) \ge {{n+1} \choose 2}$.
The total number of paths of length at most $k$ can be estimated as
\[ P_k \le {{n+1} \choose 2} + k n + 2 k m. \]
This counts paths between the neighbors of $v$, paths between a neighbor of $v$ and one of the $k$ vertices of the path nearest to $v$,
and at most $k$ vertices for each of the $m$ vertices on the path.
The number of paths going through $v$ is $P_d(v) = {{n+1} \choose 2} + m n$, counting paths between the leaf neighbors of $v$, resp., between
a leaf neighbor of $v$ and a vertex on the path. Thus
\[ \frac{BC_k(v)}{BC_d(v)} = \frac{P_k(v) \cdot P_d}{P_k \cdot P_d(v)} \ge
\frac{{{n+1} \choose 2} {{m + n + 1} \choose 2}}{({{n+1} \choose 2} + k n + 2 k m)({{n+1} \choose 2} + m n)}.      \]
Now let $n = \delta \cdot m$ and assume $k \le \delta^2 \cdot m$.
The numerator is $\Theta(\delta^2 (1 + \delta)^2 m^4) = \Theta(\delta^2 m^2)$.
The denominator is upper bounded by substituting the maximal value for $k$, and then we get
\[ \Theta( (\delta^2 m^2 + 2 \delta^3 m^2 + 2 \delta^2 m^2) \cdot (\delta^2 m^2 + 2 \delta m^2)) = \Theta(\delta^3 m^4). \]
Thus choosing a sufficiently small $\delta$ the statement follows.
\end{proof}





\section{Scale-free random trees} \label{sec:sca}
A scale-free random tree (also called a preferential attachment tree, corresponding to the case $m=1$ in the
general preferential attachment model) is formed by adding vertices one at a time. In the beginning there is a single vertex with degree 1 (corresponding to a ``virtual edge'').
The new vertex will be connected to one old vertex, and the probability of being connected to a particular old vertex is proportional to the current degree
of that vertex (thus, with $t$ vertices in the tree already, the sum of the degrees is $2 t - 1$).
We will call the vertices $\{1,\ldots,n\}$. Vertex $v$ is earlier than vertex $w$ when $v<w$, so that if $v<w$ then $v$ was added to the tree earlier than $w$.
Call the tree formed after $n$ vertices have been added $T_n$.

Bollob\'as and Riordan~\cite{bollobas2004shortest} give asymptotic results for the distribution of length-$k$ paths in the scale-free trees, including that all but $o(n^2)$ of the paths are of length $\Theta(\log n)$.  Their work is based on a formula for the probability of the scale-free tree $T_n$ containing a given
subgraph~\cite{bornholdt}. If the subgraph is a path then the formula simplifies~\cite{bollobas2004shortest}.

Let $P$ be a path in $T_n$ and let $a$ and $b$ be the endpoints of $P$, assuming, without loss of generality, that $a<b$.  Let $c$ be the least common ancestor of $a$ and $b$ in the path, where $a=c$ is possible (where the least common ancestor is the least vertex $c$ such that $a<c$ and $b<c$).  Finally, let $L$ be the set of vertices $v$ of $P$ such that $c<v<a$ and $R$ the set with $a<v<b$.
Thus the set of interior vertices of $P$ is $\{c\} \cup L \cup R$. The tuple
$(a,b,c,L,R)$ does not completely determine the path:  an element of $L$ may either be on the path from $a$ to $c$ or from $c$ to $b$.  Nevertheless, the probability that the path is in $T_n$ depends only on $(a,b,c,L,R)$.

\begin{lem}[Bollob\'as and Riordan~\cite{bollobas2004shortest}]\label{lem:path_prob}
Given a path $P$, let $a,b,c,L$, and $R$ be as above.  Then the probability that $P$ is in $T_n$, written $q(a,b,c,L,R)$, is

\[q(a,b,c,L,R) = \left\{\begin{array}{rc}\frac{2}{2c-1}\cdot\frac{1}{2b-2}\prod_{i\in L}\frac{1}{2i-1}\prod_{i\in R}\frac{1}{2i-2}\prod_{t = a+1}^{b}\frac{2t-2}{2t-3}&\text{ if }a\neq c\\
\frac{1}{2b-2}\prod_{i\in R}\frac{1}{2i-2}\prod_{t = a+1}^{b}\frac{2t-2}{2t-3}&\text{ if }a = c.
\end{array}\right.\]
\end{lem}

We use this lemma to show that the sequence of the expected values of the $k$-betweeness centralities in scale-free random trees do not cross.

\begin{theorem} \label{thm:pref8}
Let $v$ and $w$ be nodes in $T_n$ such that $v<w$.  Then for all $2\le k < n$ it holds that \[\mathbb{E}[BC_k(v)] > \mathbb{E}[BC_k(w)].\]
\end{theorem}

\begin{proof}
The normalization factor for both centralities is the same, so $BC_k(v) > BC_k(v+1)$ if and only if $P_k(v) > P_k(v+1)$.
Therefore it suffices to show that
\begin{equation} \label{eq:prob4}
\mathbb{E}[p_k(v)] > \mathbb{E}[p_k(v+1)].
\end{equation}
The case where $v+1=n$ is trivial, so assume $0\le v\le n-2$.
It holds that $\mathbb{E}[p_k(v)] = \sum_{P \ni v} p(a,b,c,L,R)$,
where the sum is over all possible paths of length $k$ containing $v$ as an interior vertex.
Inequality (\ref{eq:prob4}) is proved by constructing an injective map $f$ from the set of paths of length $k$ containing $v+1$ as an interior vertex into the set
of paths of length $k$ containing $v$ as an interior vertex such that for every path $P$ it holds that
$\mathbb{P}[f(P)\in T_n] \ge \mathbb{P}[P\in T_n]$ and for some paths the inequality is strict.


Consider a path $P$ containing $v+1$ as an interior vertex. If $v$ is an interior vertex of $P$ as well then $f(P) = P$.
If $P$ is not a vertex of $P$ then $f(P)$ is obtained by replacing $v+1$ with $v$, i.e., neighbors of $v+1$ become neighbors of $v$ and $v+1$ is deleted.
The only remaining case to consider is when $v$ is an endpoint of $P$. It has to be the case then that $v = a$, and $v+1$  belongs to the subpath between $c$ and $b$.
In order to construct $f(P)$, we delete $v+1$ from the subpath, and add it as a parent of $v$. Thus $v+1$ becomes the vertex $a$ in $f(P)$ and $v$ becomes an interior vertex
on the path from $c$ to $a$ in $f(P)$. Note that the in-degree of $c$ in $f(P)$ is 2, even if it had in-degree $1$ in $P$.


We now define $f$ more formally.  In view of (\ref{eq:prob4}), it suffices to show how $f$ acts on the tuple $(a,b,c,L,R)$ corresponding to the path $P$:

\[f(a,b,c,L,R) := \left\{\begin{array}{ll}
(a,b,c,L,R)						&\text{if } v\text{ is interior in }P, \\
(a,b,v,L,R)						&\text{if }v\not\in P\text{ and }v+1=c,\\
(a,b,c,L,R\cup\{v\}\setminus\{v+1\}) 	&\text{if }v\not\in P\text{ and }v+1 \in R,\\
(a,b,c,L\cup\{v\}\setminus\{v+1\},R)	&\text{if }v\not\in P\text{ and }v+1 \in L,\\
(v+1,b,v,\emptyset,R\setminus\{v+1\})	&\text{if } v=a=c, \\
(v+1,b,c,L\cup\{v\},R\setminus\{v+1\})	&\text{if } v=a\neq c.
\end{array}\right.\]
Abusing notation slightly, we will then define $f(P)$ as the natural path that corresponds with the tuple $f(a,b,c,L,R)$, where a given vertex in $L$ under $P$ stays on the subpath between $a$ and $b$ in $f(P)$ and vice versa (and a vertex $v$ added to $L$ stays on the side corresponding to the removed $v+1$).  Note that not all of these six cases (which we will call (1) through (6) in descending order) exist for all values of $v$ and $k$.

For the injectivity of $f$ note that the mapping preserves the path length. It follows directly from the definition that the mapping is injective within each
case of the definition. Paths obtained in the different cases have (1) $v+1$ as an interior vertex, (2), (3), (4) do not contain $v+1$, (5), (6) contain $v+1$ as
an endpoint. Cases (2), (3), (4) are distinguished by $v$ being the root, in $R$ and in $L$. Cases (5), (6) are distinguished by $v+1$ being a parent of the root or not.


It remains to show that $\mathbb{P}[f(P)\in T_n] \ge \mathbb{P}[P\in T_n]$ and inequality holds in some cases.  The intuition is that under $f$ the set of vertices in the path stay the same or get smaller, e.g. $v+1$ is removed in favor of $v$, and that a path whose vertices are earlier is more likely.
The relation between $\mathbb{P}[f(P)\in T_n]$ and $\mathbb{P}[P\in T_n]$ can be calculated explicitly using Lemma~\ref{lem:path_prob}.

 For example, in the second case ($v\not\in P$ and $v+1=c$), Lemma~\ref{lem:path_prob} gives
\[ q(a,b,v,L,R) = \frac{2v+1}{2v-1}q(a,b,v+1,L,R) > q(a,b,v+1,L,R). \]

 In the last case ($v=a\neq c$), again using Lemma~\ref{lem:path_prob} we get
\begin{eqnarray}
&&q(v+1,b,c,L\cup\{v\},R\setminus\{v+1\}) \nonumber \\
&=& \frac{2}{2c-1}\cdot\frac{1}{2b-2}\cdot\frac{2v}{2v-1}\prod_{i\in L}\frac{1}{2i-1}\prod_{i\in R}\frac{1}{2i-2}\prod_{t = a+2}^{b}\frac{2t-2}{2t-3}
= q(v,b,c,L,R). \nonumber
\end{eqnarray}

The other cases are similar:

 \[\left\{
\begin{array}{lr}
q(a,b,c,L,R) = q(a,b,c,L,R) 								&\text{if } v\text{ is interior in }P,\\
q(a,b,v,L,R)  = \frac{2v+1}{2v-1}q(a,b,c,L,R)					&\text{if }v\not\in P\text{ and }v+1=c,\\
q(a,b,c,L,R\cup\{v\}\setminus\{v+1\})= \frac{2v}{2v-2}q(a,b,c,L,R)	&\text{if }v\not\in P\text{ and }v+1 \in R,\\
q(a,b,c,L\cup\{v\}\setminus\{v+1\},R) = \frac{2v+1}{2v-1}q(a,b,c,L,R)	&\text{if }v\not\in P\text{ and }v+1 \in L,\\
q(v+1,b,v,\emptyset,R\setminus\{v+1\}) = 2\cdot q(a,b,c,L,R)		&\text{if } v=a=c,\\
q(v+1,b,c,L\cup\{v\},R\setminus\{v+1\}) = 	q(a,b,c,L,R)			&\text{if } v=a\neq c.
\end{array}
\right. \]

Thus $\mathbb{E}[p_k(v)] \ge \mathbb{E}[p_k(v+1)]$.
Strict inequality comes from noticing that in some cases $\mathbb{P}[f(P)\in T_n] > \mathbb{P}[P\in T_n]$ (namely the cases where $v\not\in P$ or $v=a=c$) or noticing that $f$ is not surjective.  However, some of these cases do not always exist for all values of $v$ and $k$.
If $v\le n-k$, then certainly the case where $v=a=c$ exists.  This implies this case exists for $v=1$ because $k\le n-1$.  Otherwise, for $2 \le v \le n-2$, we show that $f$ is not surjective:  Consider the path $n, w_1,\ldots w_s,v,w_{s+1}, \ldots, w_{k-3},1,v+1$ for any appropriate $s$.  Here $a=v+1$, $b=n$, and $c=1$.  This path cannot be in the image of $f$:  $v+1$ is in this path, but it is not an interior vertex, $v\neq c$, and $L=\emptyset$.  On the other hand, no path $f(P)$ satisfies each of these properties, but this path has $v$ as an interior vertex, completing the proof.
\end{proof}

\section{Dips} \label{sec:dip}

In this section we show that profiles in trees can be arbitrarily far from monotonic, in the sense that a profile may have many dips.
One difficulty in proving such a result is that a large number of length-$k$ paths not containing a vertex $v$ - which contributes to a decrease in its
$k$-betweenness centrality compared to previous values - affects the normalization terms for other $\ell$-betweenness centralities.  This means constructing a tree so that a given dip occurs may interfere with $v$ having other dips in its profile. The following theorem shows that this barrier is not insurmountable.


\begin{theorem}\label{thm:dips}
For any integer $i\ge 1$, there is a tree $G$ of diameter $\Theta(i)$ and a vertex $v$ in $G$ such that the profile of $v$ has $\Omega(i)$ dips.
\end{theorem}

\begin{figure}[t]
\centering
\tikzstyle{vertex}=[circle,fill=black,minimum size=10pt,inner sep=0pt]
\tikzstyle{edge}=[-,color=black,line width=1.5]
\tikzstyle{edgelabel}=[color=black]
\begin{tikzpicture}[scale=0.75, auto]
    \node[vertex] (c1) at (0,0) {};
    \node[vertex] (v) at (1,0) [label=below:$v$]{}
        edge[edge] node {} (c1);
    \node[vertex] (c2) at (2,0) {}
        edge[edge] node {} (v);
    \node[vertex] (c3) at (3,0) {}
        edge[edge] node {} (c2);

    \node[vertex] (ba1) at (3,1) {}
        edge[edge] node {} (c3);
    \node[vertex] (ba2) at (3,2) {};

    \node[vertex] (bb1) at (3,-1) {}
        edge[edge] node {} (c3);
    \node[vertex] (bb2) at (3,-2) {};

    \node[vertex] (c4) at (4,0) {}
    	edge[edge] node {} (c3);
    \node[vertex] (c5) at (5,0) {};
    \node[vertex] (c6) at (6,0) {}
        edge[edge] node {} (c5);

    \node[vertex] (ba3) at (6,1) {}
        edge[edge] node {} (c6);
    \node[vertex] (ba4) at (6,2) {};

    \node[vertex] (bb3) at (6,-1) {}
        edge[edge] node {} (c6);
    \node[vertex] (bb4) at (6,-2) {};

    \node[vertex] (c7) at (7,0) {}
    	edge[edge] node {} (c6);
    \node[vertex] (c8) at (8,0) {};
    \node[vertex] (c9) at (10,0) {};
    \node[vertex] (c10) at (11,0) {}
    	edge[edge] node {} (c9);
	
    \node[vertex] (ba5) at (11,1) {}
        edge[edge] node {} (c10);
    \node[vertex] (ba6) at (11,2) {};

    \node[vertex] (bb5) at (11,-1) {}
        edge[edge] node {} (c10);
    \node[vertex] (bb6) at (11,-2) {};

    \node[vertex] (c11) at (12,0) {}
    	edge[edge] node {} (c10);
    \node[vertex] (c12) at (13,0) {};

    \draw [decorate,decoration={brace,amplitude=15pt},xshift=0pt,yshift=8pt, line width = 1]
    (2,2) -- (5,2) node [black,midway,yshift=15pt] {$1$};
    \draw [decorate,decoration={brace,amplitude=15pt},xshift=0pt,yshift=8pt, line width = 1]
    (5,2) -- (8,2) node [black,midway,yshift=15pt] {$2$};
    \draw [decorate,decoration={brace,amplitude=15pt},xshift=0pt,yshift=8pt, line width = 1]
    (10,2) -- (13,2) node [black,midway,yshift=15pt] {$i$};

    \draw [decorate,decoration={brace,amplitude=8pt,},xshift=-8pt,yshift=2pt, line width = 1]
    (3,0) -- (3,2) node [black,midway,xshift=-8pt] {$j$};
    \draw [decorate,decoration={brace,amplitude=8pt,},xshift=-8pt,yshift=-2pt, line width = 1]
    (3,-2) -- (3,0) node [black,midway,xshift=-8pt] {$j$};
    \draw [decorate,decoration={brace,amplitude=8pt},xshift=2pt,yshift=8pt, line width = 1]
    (3,0) -- (5,0) node [black,midway,yshift=8pt] {$j$};

    \node[draw=none] at (3,1.6) {$\vdots$};
    \node[draw=none] at (3,-1.4) {$\vdots$};
    \node[draw=none] at (4.525,0) {$\cdots$};
    \node[draw=none] at (6,1.6) {$\vdots$};
    \node[draw=none] at (6,-1.4) {$\vdots$};
    \node[draw=none] at (7.525,0) {$\cdots$};
    \node[draw=none] at (9,0) {$\cdots$};
    \node[draw=none] at (9,2.5) {$\cdots$};
    \node[draw=none] at (11,1.6) {$\vdots$};
    \node[draw=none] at (11,-1.4) {$\vdots$};
    \node[draw=none] at (12.525,0) {$\cdots$};

\end{tikzpicture}
\caption{$G_{ij}$}
\label{fig:tree}
\end{figure}

\begin{proof}[Outline]
For any $i,j>0$ let $G_{ij}$ be the tree shown in Figure~\ref{fig:tree}.
It consists of a `central' path of length $i(j+1)+2$, with $2i$ simple paths adjoined to the central path, each of length $j$ (which we will refer to as the `branches').  Two of these paths are adjoined to vertex $\ell$ on the central path, for every $\ell=3 \mod j+1$ (where $\ell=0$ refers to one of the leaves of the central path, $\ell=1$ its neighbor, etc.)  The tree has diameter $ij+i+j-1$ (as long as $j\ge 3$).
The vertex that we show has $\Omega(i)$ dips is the vertex $\ell=1$.  Call this vertex $v$.


We show that for an appropriately chosen fixed constant $j$ the following holds for $2 \le r \le i-1$:
\begin{equation}\label{eq:dip_location}
BC_{r(j+1)+2}(v) > BC_{r(j+1)+3}(v) < BC_{r(j+1)+4}(v).
\end{equation}
This yields $i-1$ dips in the profile, which suffices.

It is straightforward to calculate the number of paths through $v$, because all paths through $v$ start at the vertex $l=0$ and end at each of the other vertices except for $v$ itself:  $P_{r(j+1)+2}(v)=r(3j+1)+1$, $P_{r(j+1)+3}(v)=r(3j+1)+2$, and $P_{r(j+1)+4}(v)=r(3j+1)+5$.


We omit the calculations for $P_k$, which are calculated in a manner similar to the manner $P_k$ is calculated in the proof of Lemma~\ref{lem:simple_path_form}.  The number of such length-$k$ paths is counted by counting the number of paths that are  1) along the central path, 2) on two of the branches, or 3) on only one of the branches.  By direct calculation, we get that the values of $P_k$ needed are as given by Table~\ref{table:P_k_dips} when $j\ge 5$.

This explicit calculation of the values of $P_k$ and $P_k(v)$ yields that (\ref{eq:dip_location}) holds when $j=5$  for all $2 \le r \le i-1$.
\end{proof}

A more complete enumeration of the calculations may be found in the appendix.

\begin{table}[t]
\caption{Values of $P_k$ for $G_{ij}$}
\begin{center}
\begin{tabular}{| l | p{8.5cm} |}
\hline
$\boldsymbol{k}$ & $\boldsymbol{P_k}$ \\ \hline
$r(j+1)+2$ & $r^2(-\frac{9}{2} j^2-3j-\frac{1}{2})+r(9ij^2 + 6ij + i +6j^2 -\frac{23}{2}j +\frac{17}{2}) + i(-6j^2+16j-7)-3j^2+15j-17$ \\\hline
$r(j+1)+3$ & $P_{r(j+1)+2} + (9j-3)(i-r) +6j-18$\\\hline
$r(j+1)+4$ & $P_{r(j+1)+3} + (9j-3)(i-r) +6j-25$\\\hline
\end{tabular}
\end{center}
\label{table:P_k_dips}
\end{table}

Note that the number of dips in this proof is not only asymptotically as large as possible in the diameter, but also in the number of vertices, since these trees have $\Theta(i)$ vertices.

\section{Crossings} \label{sec:cro}

In this section we analyze the number of crossings between two profiles. The number of crossings of two profiles does not depend on whether normalization is
used or not - thus examining crossings avoids the issue of normalization.



\begin{figure}
\tikzstyle{vertex}=[circle,fill=black,minimum size=10pt,inner sep=0pt]
\tikzstyle{edge}=[-,color=black,line width=1.5]
\tikzstyle{edgelabel}=[color=black]
\begin{tikzpicture}[scale=.78, auto]

    \node[vertex] (u) at (0,0) [label=above left:$u$] {};
    \node[vertex] (b) at (1,0) {}
        edge[edge] node {} (u);
    \node[vertex] (c) at (2,0) {}
        edge[edge] node {} (b);
    \node[vertex] (v) at (4,0) [label=above right:$v$] {};

    \node[vertex] (u1) at (1,-1) {}
    	edge[edge] node {} (u);
    \node[vertex] (u2) at (-1,-1) {}
    	edge[edge] node {} (u);
    \node[vertex] (u3) at (-3,-1) {}
    	edge[edge] node {} (u);
    \node[vertex] (u4) at (-4,-1) {}
    	edge[edge] node {} (u);
	
   \node[vertex] (a01) at (-4,-0.6) {}
   	edge[edge] node {} (u);
   \node[vertex] (a02) at (-4,0.2) {}
   	edge[edge] node {} (u);
	
    \node[vertex] (u5) at (1,-2) {}
    	edge[edge] node {} (u1);
    \node[vertex] (u6) at (-1,-2) {}
    	edge[edge] node {} (u2);
    \node[vertex] (u7) at (-3,-2) {}
    	edge[edge] node {} (u3);
	
    \node[vertex] (u9) at (1,-3) {}
    	edge[edge] node {} (u5);
    \node[vertex] (u10) at (-1,-3) {}
    	edge[edge] node {} (u6);
    \node[vertex] (u11) at (-3,-3) {}
    	edge[edge] node {} (u7);
	
    \node[vertex] (a10) at (-2.6,-3) {}
    	edge[edge] node {} (u7);
    \node[vertex] (a11) at (-1.8,-3) {}
    	edge[edge] node {} (u7);
	
    \node[vertex] (u13) at (1,-4) {}
    	edge[edge] node {} (u9);
    \node[vertex] (u14) at (-1,-4) {};
	
    \node[vertex] (u17) at (1,-5) {};
    \node[vertex] (u18) at (-1,-5) {}
    	edge[edge] node {} (u14);

    \node[vertex] (ai0) at (-0.6,-5) {}
    	edge[edge] node {} (u14);
    \node[vertex] (ai1) at (0.2,-5) {}
    	edge[edge] node {} (u14);

    \node[vertex] (u19) at (1,-6) {}
    	edge[edge] node {} (u17);
	
    \node[vertex] (al0) at (1.4,-6) {}
    	edge[edge] node {} (u17);
    \node[vertex] (al1) at (2.2,-6) {}
    	edge[edge] node {} (u17);
	
    \node[vertex] (v1) at (3,-1) {}
    	edge[edge] node {} (v);
    \node[vertex] (v2) at (5,-1) {}
    	edge[edge] node {} (v);
    \node[vertex] (v3) at (7,-1) {}
    	edge[edge] node {} (v);
    \node[vertex] (v4) at (9,-1) {}
    	edge[edge] node {} (v);
	
    \node[vertex] (v5) at (3,-2) {}
    	edge[edge] node {} (v1);
    \node[vertex] (v6) at (5,-2) {}
    	edge[edge] node {} (v2);
    \node[vertex] (v7) at (7,-2) {}
    	edge[edge] node {} (v3);
    \node[vertex] (v8) at (9,-2) {}
    	edge[edge] node {} (v4);
	
    \node[vertex] (b00) at (3.4,-2) {}
    	edge[edge] node {} (v1);
    \node[vertex] (b01) at (4.2,-2) {}
    	edge[edge] node {} (v1);
	
    \node[vertex] (v9) at (5,-3) {}
    	edge[edge] node {} (v6);
    \node[vertex] (v10) at (7,-3) {}
    	edge[edge] node {} (v7);
    \node[vertex] (v11) at (9,-3) {}
    	edge[edge] node {} (v8);
	
    \node[vertex] (v12) at (5,-4) {}
    	edge[edge] node {} (v9);
    \node[vertex] (v13) at (7,-4) {}
    	edge[edge] node {} (v10);
    \node[vertex] (v14) at (9,-4) {}
    	edge[edge] node {} (v11);
	
    \node[vertex] (b10) at (5.4,-4) {}
    	edge[edge] node {} (v9);
    \node[vertex] (b11) at (6.2,-4) {}
    	edge[edge] node {} (v9);
	
    \node[vertex] (v15) at (7,-5) {};
    \node[vertex] (v16) at (9,-5) {}
    	edge[edge] node {} (v11);
	
    \node[vertex] (v17) at (7,-6) {}
        edge[edge] node {} (v15);
    \node[vertex] (v18) at (9,-6) {};

    \node[vertex] (bi0) at (7.4,-6) {}
    	edge[edge] node {} (v15);
    \node[vertex] (bi1) at (8.2,-6) {}
    	edge[edge] node {} (v15);
	
    \node[vertex] (v19) at (9,-7) {}
        edge[edge] node {} (v18);

    \node[vertex] (bl0) at (9.4,-7) {}
    	edge[edge] node {} (v18);
    \node[vertex] (bl1) at (10.2,-7) {}
    	edge[edge] node {} (v18);

    \draw [decorate,decoration={brace,amplitude=15pt},xshift=0pt,yshift=8pt, line width = 1]
    (0,0) -- (4,0) node [black,midway,yshift=15pt] {$2\ell$};
    \draw [decorate,decoration={brace,amplitude=10pt},xshift=8pt,yshift=0pt, line width = 1]
    (-1,-1) -- (-1,-4) node [black,midway,xshift=10pt] {$2i-1$};
     \draw [decorate,decoration={brace,amplitude=10pt},xshift=8pt,yshift=0pt, line width = 1]
    (1,-1) -- (1,-5) node [black,midway,xshift=10pt] {$2\ell-3$};
     \draw [decorate,decoration={brace,amplitude=8pt},xshift=-8pt,yshift=0pt, line width = 1]
    (-4,-1) -- (-4,0.2) node [black,midway,xshift=-8pt] {$A_0$};
     \draw [decorate,decoration={brace,amplitude=8pt},xshift=0pt,yshift=-8pt, line width = 1]
    (-1.8,-3) -- (-3,-3) node [black,midway,yshift=-8pt] {$A_1$};
     \draw [decorate,decoration={brace,amplitude=8pt,mirror},xshift=0pt,yshift=-8pt, line width = 1]
    (-1,-5) -- (0.2,-5) node [black,midway,yshift=-22pt] {$A_i$};
     \draw [decorate,decoration={brace,amplitude=8pt,mirror},xshift=0pt,yshift=-8pt, line width = 1]
    (1,-6) -- (2.2,-6) node [black,midway,yshift=-22pt] {$A_{\ell-1}$};
    \draw [decorate,decoration={brace,amplitude=8pt,mirror},xshift=0pt,yshift=-8pt, line width = 1]
    (3,-2) -- (4.2,-2) node [black,midway,yshift=-22pt] {$B_0$};
     \draw [decorate,decoration={brace,amplitude=8pt,mirror},xshift=0pt,yshift=-8pt, line width = 1]
    (5,-4) -- (6.2,-4) node [black,midway,yshift=-22pt] {$B_1$};
     \draw [decorate,decoration={brace,amplitude=8pt,mirror},xshift=0pt,yshift=-8pt, line width = 1]
    (7,-6) -- (8.2,-6) node [black,midway,yshift=-22pt] {$B_i$};
    \draw [decorate,decoration={brace,amplitude=10pt},xshift=8pt,yshift=0pt, line width = 1]
    (7,-1) -- (7,-5) node [black,midway,xshift=10pt] {$2i$};
    \draw [decorate,decoration={brace,amplitude=10pt},xshift=8pt,yshift=0pt, line width = 1]
    (9,-1) -- (9,-6) node [black,midway,xshift=10pt] {$2\ell-2$};
     \draw [decorate,decoration={brace,amplitude=8pt,mirror},xshift=0pt,yshift=-8pt, line width = 1]
    (9,-7) -- (10.2,-7) node [black,midway,yshift=-22pt] {$B_{\ell-1}$};

    \node[draw=none]  at (3,0) {$\cdots$};
    \node[draw=none] at (0,-1) {$\cdots$};
    \node[draw=none] at (-2,-1) {$\cdots$};
    \node[draw=none] at (-4,-0.1) {$\vdots$};
    \node[draw=none] at (-2.2,-3) {$\cdots$};
    \node[draw=none] at (1,-4.4) {$\vdots$};
    \node[draw=none] at (-1,-3.4) {$\vdots$};
    \node[draw=none] at (-0.2,-5) {$\cdots$};
    \node[draw=none] at (1.8,-6) {$\cdots$};
    \node[draw=none] at (6,-1) {$\cdots$};
    \node[draw=none] at (8,-1) {$\cdots$};
    \node[draw=none] at (3.8,-2) {$\cdots$};
    \node[draw=none] at (5.8,-4) {$\cdots$};
    \node[draw=none] at (7,-4.4) {$\vdots$};
    \node[draw=none] at (7.8,-6) {$\cdots$};
    \node[draw=none] at (9.8,-7) {$\cdots$};
    \node[draw=none] at (9,-5.4) {$\vdots$};

\end{tikzpicture}
\vskip -0.1 in
\caption{$T_{\ell}$}
\label{fig:two_profile_graph}
\end{figure}

\begin{theorem}\label{thm:crossings}
For any integer $\ell\ge 1$, there is a tree $T$ of diameter $\Theta(\ell)$ with vertices $u$ and $v$ whose profiles cross $\Omega(\ell)$ times.

\end{theorem}

\begin{proof}
For each $\ell\ge 1$, we build a tree $T_\ell$ with vertices $u$ and $v$ as shown in Figure~\ref{fig:two_profile_graph}, where the number of leaves in the sets $A_i$ and $B_i$ will be determined shortly.  This tree consists of a simple path of length $2\ell$ between $u$ and $v$ and $\ell$ brooms attached to each.  The tree is constructed so that paths of length $2i$ that include $u$ (as an interior vertex) can reach the set $A_{i-1}$ of vertices as depicted (as exterior vertices), but that paths of length $2i$ through $v$ cannot reach the set $B_{i-1}$ of vertices.  Let $a_i = |A_i|$ and $b_i=|B_i|$.  We then choose $a_{i-1}$ so that $BC_{2i}(u) > BC_{2i}(v)$, and similarly we choose $b_{i-1}$ so that $BC_{2i+1}(v) > BC_{2i+1}(u)$.  This gives us $2\ell-1$ crossings when $1\le i \le \ell-1$.

We will start, for $2\le k\le 2\ell+1$, by upper and lower bounding $p_k(u)$ and $p_k(v)$ as a function of the $a_i$'s and $b_i$'s.  Certainly, since the path between $u$ and $v$ is length $2\ell$, all paths containing $u$ do not contain $v$ and vice-versa (as internal nodes).  So we can upper bound $p_k(u)$ as the square of the number of vertices within distance $k-1$ of $u$ on both the path between $u$ and $v$ and the subtree of $u$, ignoring that paths can't start and end at the same vertex (this approximation will be easily sufficient for our needs.)  We show the following holds:

\begin{equation}\label{eq:p_k_u}
a_{\lfloor\frac{k}{2}\rfloor-1} < p_{k}(u) <  \left(\sum_{j=0}^{\lfloor\frac{k}{2}\rfloor-1} a_j +(\ell+1)(k-1)\right)^2.
\end{equation}

Since the leaves of the $a_{\lfloor\frac{k}{2}\rfloor-1}$ edges as depicted are length no more than $k-1$ from $u$, there is at least one path through $u$ of length $k$, so certainly $p_k(u) > a_{\lfloor\frac{k}{2}\rfloor-1}$.  On the subtree of $u$, at distance $j$ from $u$, there are either $a_j + k-1$ vertices or $k$ vertices, depending on the parity of $j$.  Thus there are no more than a total of $\sum_{j=0}^{\lfloor\frac{k}{2}\rfloor-1} a_j +\ell(k-1)$ vertices.  Including the additional $k-1$ vertices on the path between $u$ and $v$, this yields an upper bound on the number of paths of length no more than $k$ containing $u$ as given on the right-hand side of equation~\eqref{eq:p_k_u}.  Similarly, we get the following for $p_k(v)$:

\begin{equation}\label{eq:p_k_v}
b_{\lfloor\frac{k}{2}\rfloor-1} < p_{k}(v) <  \left(\sum_{j=0}^{\lceil\frac{k}{2}\rceil-2} b_j +(\ell+1)(k-1)\right)^2.
\end{equation}

Now we choose each $a_i$ and $b_i$ inductively, ensuring that $BC_{2i}(u) > BC_{2i}(v)$ and $BC_{2i+1}(v) > BC_{2i+1}(u)$.  In the base case, for $k=2$, where we will define $a_{\frac{k}{2}-1}=a_0$, we will need $BC_{2}(u) > BC_2(v)$.  $BC_2(u) = p_2(u) > a_k$ and the degree of $v$ is $\ell+1$, so it suffices to set $a_k = \binom{\ell+1}{2}$.

For arbitrary even $k$, $a_0,\ldots,a_{\frac{k}{2}-2}$ and $b_0,\ldots,b_{\frac{k}{2}-2}$ are already defined.
We want to show that $BC_k(u) > BC_k(v)$.  $BC_k(u) = \sum_{i=2}^k p_i(u) > 2\left(\sum_{j=0}^{\frac{k}{2}-2} a_{j}\right) + a_{\frac{k}{2}-1}$.  $BC_k(v) = \sum_{i=2}^k p_i(v)$,  and from the right-hand side of~\eqref{eq:p_k_v}, it suffices to notice that each $p_i(v)$ from this sum is upper-bounded by a function of $\ell$ and $b_0,\ldots,b_{\frac{k}{2}-2}$; for convenience call the sum over all $p_i(v)$ of these upper bounds $s_k$.  Thus it suffices to set $a_{\frac{k}{2}-1}$ such that \[a_{\frac{k}{2}-1} \ge s_k - 2\sum_{j=0}^{\frac{k}{2}-2} a_{j}.\]

Similarly, for arbitrary odd $k$, $a_0,\ldots,a_{\frac{k-1}{2}-1}$ and $b_0,\ldots,b_{\frac{k-1}{2}-2}$ are already defined and we want to show that $BC_k(u) < BC_k(v)$.  And similarly, we know that $BC_k(v) > 2\left(\sum_{j=0}^{\frac{k-1}{2}-1} b_{j}\right)$ and $BC_k(u)$ is upper-bounded by a function $s'_k$ of $\ell$ and $a_0,\ldots,a_{\frac{k-1}{2}-1}$.  So we set $b_{\frac{k-1}{2}-1}$ such that
\[b_{\frac{k-1}{2}-1} \ge \frac{1}{2}\left(s'_k - 2\sum_{j=0}^{\frac{k-1}{2}-2} b_{j}\right).\]

Hence constructing $a_i$ and $b_i$ for all $i\le\ell-1$ yields the desired $2\ell-1$ crossings.

\end{proof}

The bound of $\Omega(\ell)$ crossings is worst possible in order of magnitude as a function of the diameter, but unlike in the construction in Theorem~\ref{thm:dips}, the diameter is not linear in the number of vertices $n$.
  Indeed, $T_\ell$ has $\tilde{O}(\ell^{2^\ell})$ vertices, which yields $\Omega(\log \log n)$ crossings.  We leave for future work if it is possible to increase the number of crossings as a function of $n$.

\section{Experiments} \label{sec:expe}

Theorem~\ref{thm:pref8} shows that the profiles of expected betweenness centralities do not cross, but it does not imply results on the expected number of crossings or the probability that there is at least one crossing.  Towards results of this kind, we give experimental results that give initial evidence for the behavior of profiles in scale-free random trees in terms of crossings and dips.

\begin{figure}[t]
\begin{center}
\centerline{\includegraphics[scale=0.62]{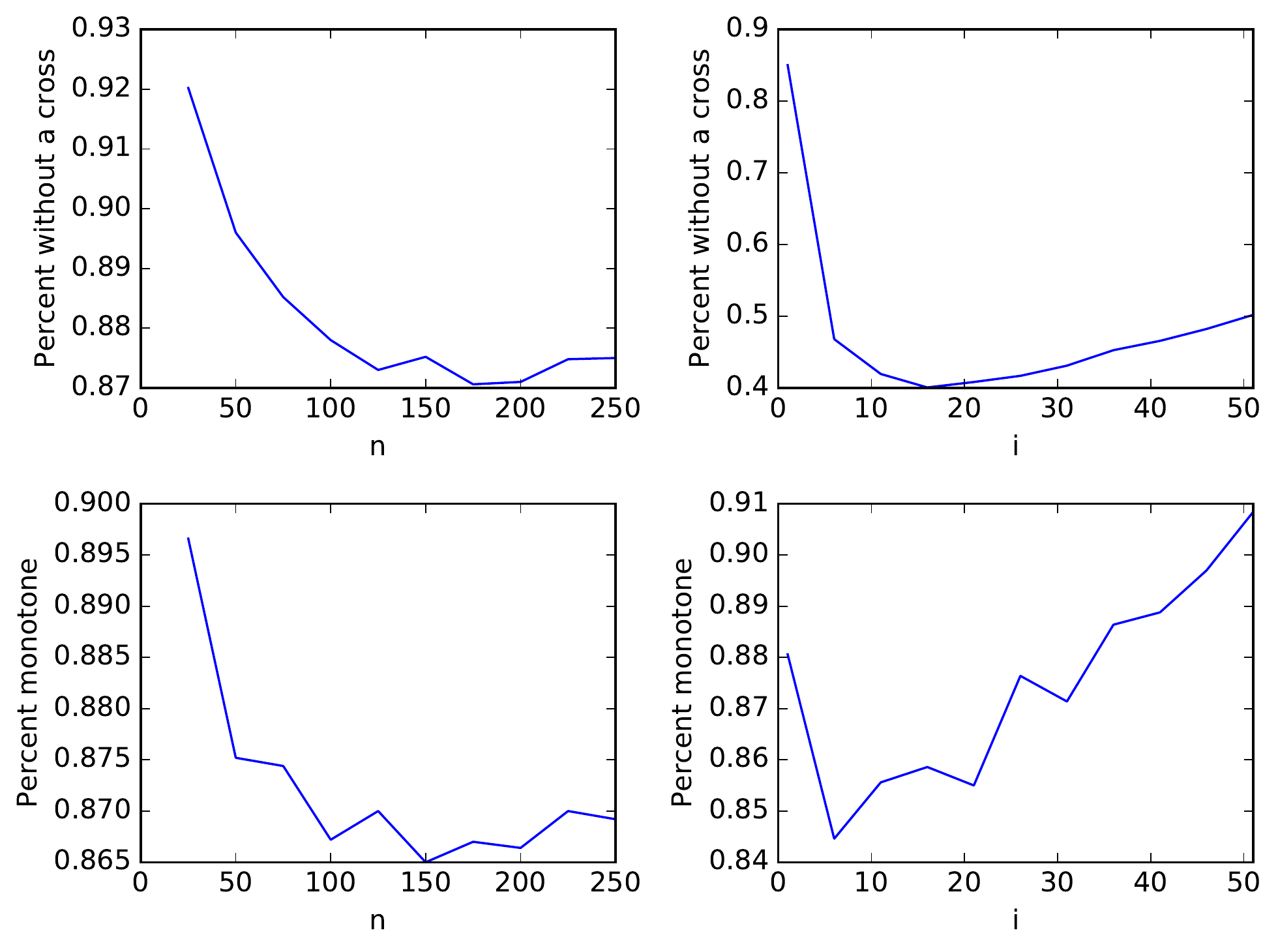}}
\end{center}
\vskip -.15in
\caption{Clockwise, from top left:  the estimated probability that $\mathcal{B}(1)$ and $\mathcal{B}(2)$ have no crossings as the number of vertices $n$ increases, the estimated probability that $\mathcal{B}(i)$ and $\mathcal{B}(i+1)$ have no crossings as $i$ increases, the estimated probability that $\mathcal{B}(i)$ is monotone, as $i$ increases, and the estimated probability that $\mathcal{B}(1)$ is monotone, as $n$ increases.}
\label{fig:plots}
\end{figure}

We give here some examples of experiments performed with scale-free trees and discuss the problems suggested by these experiments.
In Figure~\ref{fig:plots} (top left), we estimate the probability that there are no crossings between the profiles of the first two vertices.  We also show the results for fixed $n$
(we use $n=250$) but between vertices $i$ and $i+1$ as $i$ increases (top right) and results for the likelihood of monotonicity - i.e.~no dips - of a vertex (bottom).  In each experiment, each value is averaged over $5000$ independent trials.

As $n$ increases the likelihood of there being no dips in the profile of vertex $1$ (bottom left) or no crossings between the profiles of vertex $1$ and $2$ (top left) actually decreases (where small deviations are likely due to sampling error).
This is indication that profiles in large trees may be more likely to behave badly. In that context,
the worst-case results discussed in Sections~\ref{sec:dip} and~\ref{sec:cro} may provide useful information.

On the other hand, the likelihood of having no crossings and being monotonic as the vertex $i$ increases for fixed $n$ (again, we use $n=250$) looks very different (top right and bottom right of Figure~\ref{fig:plots}).  After initial decreases, the likelihood of being monotonic and not having crossings increases as $i$ increases.  This implies that having crosses and dips may be limited to earlier vertices and that vertices close to leaves do not behave similarly to our worst-case results.  We leave a more precise analysis of which vertices in scale-free trees behave badly for future work.

\section{Remarks and future work}

In this paper we studied $k$-betweenness centrality by considering properties of the betweenness centrality profiles of trees, motivated by experimental results of 
Pfeffer and Carley~\cite{pfeffer2012k}. We showed that for scale-free random trees it holds that if $v < w$ then $\mathbb{E}[BC_k(v)] > \mathbb{E}[BC_k(w)]$;
thus the profile of the expected betweenness centralities of an earlier vertex dominates the profile of the expected betweenness centralities of a later vertex.

We also show, in the worst case, that there are profiles of vertices in trees where the number of crossings and dips is proportional to the diameter, which is always an upper bound.  This means that in the worst case it suffices to consider trees; vertices' profiles in trees may behave as non-smoothly as vertices' profiles in any graph.  

Smoothness is an important property for understanding how the structure of local neighborhoods of a graph generalize to the entire graph.
However, our work demonstrates that in the worst case vertices' profiles may be very non-smooth and that smoothness assumptions will not hold in general, opening up for future work questions on establishing more precise bounds on how severe departures from smoothness may be.

For example, is it possible to construct examples where many vertices in the same tree have profiles all with many dips or crossings?  More generally, to what degree is non-smooth behavior an isolated problem where any dips or crossings are very small or short-lived?

We also leave as an open problem further results on the expected behavior of profiles in scale-free trees, such as determining the number of crosses, say, between the root and the next vertex, as $n$ goes to infinity.  Our preliminary experimental results, described in the previous section, provide some information towards this type of question.
Problems about crossings tend to be easier than problems about dips, as the former do not have to deal with a normalization term.  Our experimental results, however, give a starting point towards, say, determining the expected number of dips in a profile.

Our result for the scale-free random trees is based on the work of Bollob\'as and Riordan~\cite{bornholdt,bollobas2004shortest}; their other results in~\cite{bollobas2004shortest} may be relevant in this context. Both earlier techniques and results, e.g.~Smythe and Mahmoud~\cite{smythe1995survey}, and more recent results, e.g.~Bubeck et al.~\cite{bubeck14seed,bubeck14adam} and Jog and Loh~\cite{jog16tree}, form an emerging picture of scale-free trees and may be useful for future work.

\bibliographystyle{plain}
\bibliography{cits}


\section*{Appendix}
\appendix

\section{Dips}

In the proof of Theorem~\ref{thm:dips}, we omitted the calculation of $P_k$ for $G_{ij}$, which we return to in this section.  Specifically, we show that $p_k$, the number of paths of exactly length $k$, is as given in Table~\ref{table:p_k_dips}, for all $2\le r \le i-1$ and $j\ge 5$ (if $j < 5$ then not all of these cases are distinct):
\begin{table}[h]
\caption{Values of $p_k$ for $G_{ij}$}
\label{table:p_k_dips}
\begin{center}
\begin{tabular}{| l | p{6cm} |}
\hline
$\boldsymbol k$ & $\boldsymbol{p_k}$\\ \hline
$k=2$ & $3ij+4i+1$ \\ \hline
$k=3$ & $3ij+7i$ \\\hline
$4\le k\le j+1$ & $i(3j+3)+(3k-5)(i-1)+6$ \\\hline
$j+2\le k\le j+4$ & $k(3i-9)+i(3j-6)+6j+21$ \\\hline
$j+5\le k\le 2j$ & $k(3i-9)+i(3j-6)+6j+23$ \\\hline
$k = 2j+1$ & $k(-4i+1)+i(17j+1)-14j+13$ \\\hline
$2j+2\le k\le 2j+3$ & $k(-4i-1)+i(17j+9)-10j+9$ \\\hline
$r(j+1)+2\le k\le r(j+1)+3$ & $-9k+i(9j-3)+6j+12r+9$ \\\hline
$r(j+1)+4\le k\le (r+1)j+r-1$ & $-9k+i(9j-3)+6j+12r+11$ \\\hline
$(r+1)j+r\le k\le (r+1)j+r+1$ & $k(4i-3-4r)+i(-3+5j-4rj-4r)+j(-2r+4r^2)+6r+4r^2+11$ \\\hline
$k=(r+1)j+r+2$ & $k(-4i-3+4r)+i(4rj+13j+4r+5)+j(-4r^2-10r)-4r^2-2r+9$ \\
\hline
\end{tabular}
\end{center}
\end{table}

The number of such length-$k$ paths are counted by counting the number of paths that are 1) along the central path, 2) on two of the branches, and 3) on only one of the branches.  For brevity's sake, we will show just the cases when $k=2$ and the cases when $r(j+1)+2 \le k \le (r+1)j+r-1$.  The other cases are similar.

We use two simple facts:  There are $m-k+1$ subpaths of length $k\le m$ along a path of length $m$, and for a fixed path of length $m$, there are $k-m+1$ paths of length $k \ge m$ that contain the fixed path as a subpath.

First we show that there are $3ij+4i+1$ length $k=2$ paths.  The length of the central path is $i(j+1)+2$, so there are $i(j+1) + 2 - k+1$ such subpaths of length $k$.  Otherwise a length $k$ path must intersect either two branches that attach to the same vertex or only one of those branches.  There are $i$ pairs of branches, each inducing a path of length $2j$, so there are $i(2j-k+1)$ such paths.  Else the length $k$ path intersects only one branch.  For each branch, there are two paths that intersect that branch but do not intersect another branch:  the path can go to the `left' or `right', giving us a total of $4i$ such paths.  Summing up, we get $3ij+4i+1$ paths.

The other cases we consider are the two cases are (as given above) when $r(j+1) + 2 \le k \le (r+1)j+r-1$, for which we will need $j\ge 5$.  The length of the central path is $i(j+1)+2$, so there are $i(j+1) + 2 - k+1$ such subpaths of length $k$.  Otherwise, a path of length $k$ intersects one of the branches in at least one edge.  Paths can, of course, only intersect at most two different branches.  Since $k>2j$, these branches cannot attach to the central path at the same vertex.  Then branches are at integer multiples of $j+1$ apart.  So a length $k$ path is too short to reach between branches that are $r+1$ apart (we say the branches are 0 apart when they are attached to the same vertex, etc.) because then the vertices at which the branches attach are distance $(r+1)(j+1)$ from each other, but $k \le (r+1)j+r-1$.  Similarly, the branches cannot be as close as $r-1$ apart.  Thus there are only two cases for the paths that intersect at exactly two branches:  the branches are either $r-1$ apart or $r$ apart.  There are $4(i-(r-1))$ pairs of branches $r-1$ apart, and the path induced by the two branches and the central path is length $(r+1)j+r-1$, so there are a total of \[4((r+1)j+r-k)(i-r+1)\] length $k$ paths in this case.  (Note in this case all such paths of length $k$ automatically include all of the appropriate edges of the central path.)  Otherwise, the branches are $r$ apart.  However, now not all length $k$ paths along the path induced by the branches contain the part of the induced path along the central path.  Here we use the second simple fact as mentioned above:  the length of the central path that all the length k paths must include is $r(j+1)$, but they must include at least one vertex from each of the branches as well, so there are $k-r(j+1)+1-2$ such paths.  Now there are $4(i-r)$ pairs of branches $r$ apart, yielding a total of \[4(k-r(j+1)-1)(i-r)\] such paths.  Finally, there are the length $k$ paths that intersect only one branch.  A path starting on a branch has two directions along the central path where its end can lie.  There are $j$ vertices along a branch where the vertex can start, and $2(i-r)$ such branches if the paths goes to the `right' and $2(i-(r+1))$ going to the `left' for which no length $k$ path is too long.  There are two more branches in either direction at which length $k$ paths can start, for which it is straightforward to verify that there are $2((r+1)j+r-1-k+1)$ going to the `right' end, and

 \[\left\{
                    \begin{array}{ll}
                        2(j+r(j+1) + 3-k+1) & \text{if } k > r(j+1) + 3,\\
                        2j & \text{if } k = r(j+1) + 3,\\
                        2j+2 & \text{if } k = r(j+1) + 2\\
                    \end{array}
\right\}\]

paths going to the `left' end.  Thus we get the two cases as listed above, for $r(j+1)+2 \le k \le r(j+1) + 3$ and for $r(j+1)+4 \le k \le r(j+1) + r- 1$.

Using these numbers, we can now calculate $P_k = \sum_{\ell=2}^k p_\ell$, yielding the values given in Table~\ref{table:P_k_dips_appendix}.

\begin{table}[h]
\caption{Values of $P_k$ for $G_{ij}$}
\label{table:P_k_dips_appendix}
\begin{center}
\begin{tabular}{| l | p{8.5cm} |}
\hline
$\boldsymbol{k}$ & $\boldsymbol{P_k}$ \\ \hline
$r(j+1)+2$ & $r^2(-\frac{9}{2} j^2-3j-\frac{1}{2})+r(9ij^2 + 6ij + i +6j^2 -\frac{23}{2}j +\frac{17}{2}) + i(-6j^2+16j-7)-3j^2+15j-17$ \\\hline
$r(j+1)+3$ & $P_{r(j+1)+2} + (9j-3)(i-r) +6j-18$\\\hline
$r(j+1)+4$ & $P_{r(j+1)+3} + (9j-3)(i-r) +6j-25$\\\hline
\end{tabular}
\end{center}
\end{table}

Counting the number of paths through $v$ is much easier.  Since $v$ has only one vertex on its left, $P_k(v)$ is just the number of vertices to its right that are within distance $k$.  Specifically, this gives us

 \[\left\{
                    \begin{array}{ll}
                        P_{r(j+1)+2}(v)=&r(3j+1)+1,\\
                        P_{r(j+1)+3}(v)=&r(3j+1)+2,\\
                        P_{r(j+1)+4}(v)=&r(3j+1)+5.\\
                    \end{array}
\right. \]

Plugging in these values and the values from Table~\ref{table:P_k_dips_appendix}, we get that for $j=5$ and for all $2 \le r \le i-1$,
\[\frac{P_{r(j+1)+2}(v)}{P_{r(j+1)+2}} > \frac{P_{r(j+1)+3}(v)}{P_{r(j+1)+3}} < \frac{P_{r(j+1)+4}(v)}{P_{r(j+1)+4}}.\]

That is, $BC_{6r+2}(v) > BC_{6r+3}(v) < BC_{6r+4}(v)$ for all $2 \le r \le i-1$ in $G_{i,5}$, as needed.

Note that the number of dips we get, namely $i-2$ dips, is not only asymptotically as large as possible in the diameter, but also in the number of vertices, since the trees $G_{i,5}$ have $\Theta(i)$ vertices.


\end{document}